\documentclass[11pt]{article}
\usepackage{graphicx} 

\usepackage{algorithm}
\usepackage[noend]{algpseudocode}
\usepackage[utf8]{inputenc}
\usepackage{amsmath,amssymb,amsthm}
\usepackage{url,hyperref}
\usepackage{fancyvrb}
\usepackage[dvipsnames]{xcolor}
\usepackage{caption}
\usepackage{tcolorbox}
\usepackage{authblk}

\newcommand{\eps}{\epsilon}

\newcommand{\etal}{{\it{et al.}}}

\newcommand{\E}{\mathcal{E}}

\newcommand{\EE}{\mathbb{E}}
\newcommand{\var}{\mathrm{Var}}

\newcommand{\C}{\mathcal{C}}

\newcommand{\RR}{\mathcal{R}}

\newcommand{\wt}[1]{\widetilde{#1}}
\newcommand{\norm}[1]{\left\lVert#1\right\rVert}
\newcommand{\wpp}{\widehat{p}}

\newtheorem{thm}{Theorem}[section]
\newtheorem{lem}{Lemma}[section]
\newtheorem{fact}{Fact}[section]
\newtheorem{defn}{Definition}[section]

\newtheorem{claim}{Claim}[section]

\title{Almost Instance Optimal Sum and Moment Estimation Using Weighted Sampling}
\author{
\begin{minipage}[t]{0.31\textwidth}
    \centering
    Anup Bhattacharya\footnote{National Institute of Science Education and Research, Bhubaneswar, 752050, India} \footnote{Homi Bhabha National Institute, Training School Complex, Anushakti Nagar, Mumbai, 400094, India}\\
    \small\href{mailto:anup@niser.ac.in}{anup@niser.ac.in}
\end{minipage}
\hfill
\begin{minipage}[t]{0.31\textwidth}
    \centering
    Suryendu Mondal\textsuperscript{*\textsuperscript{\dag}}\\
    \small\href{mailto:suryendu.mondal@niser.ac.in}
    {suryendu.mondal@niser.ac.in}
\end{minipage}
\hfill
\begin{minipage}[t]{0.31\textwidth}
    \centering
    Pinki Pradhan\textsuperscript{*\textsuperscript{\dag}}\\
    \small\href{mailto:pinki.pradhan@niser.ac.in}
    {pinki.pradhan@niser.ac.in}
\end{minipage}
}
\date{}
\begin{document}
\maketitle

\begin{abstract} We design an almost instance optimal algorithm for the sum estimation problem using weighted sampling. We show that the sample complexity for sum estimation is closely related to the $\ell_2$ norm of the weighted sampling distribution. We show an almost instance optimal lower bound for this problem as well. We also study the moment estimation problem and design an algorithm that has better instance-wise sample complexity bounds. \end{abstract}

\section{Introduction} Let $A=\{a_1,a_2,\ldots,a_n\}$ be a set of $n$ elements, where element $a_i$ has a non-negative weight $w(a_i)$. Let $W=\sum_{a\in A} w(a)$ be the total weight of elements in $A$. In our computational model, one does not get direct access to the weights of the elements in $A$. Instead, one gets sample access to the elements of $A$ using an oracle. We assume access to an oracle that returns an element $a\in A$ with probability proportional to its weight $w(a)$. That is, the oracle samples element $a_i\in A$ with probability $p_i=\frac{w(a_i)}{W}$. Throughout, we refer to these samples as weighted samples or proportional samples. Our objective is to design an algorithm that computes a {\em good} estimate of $W$ using a sublinear number of queries to the oracle. Next, we formally define the {\em sum estimation} problem.

\begin{defn}($(\eps,\delta)$-sum estimation) Given sample access to a set $A$ with $n$ weighted elements, and parameters $\eps,\delta\in (0,1)$, the objective is to obtain an estimate $\hat{W}$ of $W=\sum_{a\in A} w(a)$ such that $\Pr[|\hat{W}-W|>\eps W]\leq \delta$. \end{defn}

Our objective in this work is to design an algorithm that compues an $(\eps,\delta)$-estimate of $W$ using a small number of queries to the oracle. Beretta and T{\v{e}}tek \cite{BT2024} showed that, in the worst case, obtaining an $(\eps,1/3)$-estimate of $W$ requires $\Theta(\frac{\sqrt{n}}{\eps})$ proportional samples. Their estimator crucially relies on observing a collision in the sample. Since, in the worst case, $\Omega(\sqrt{n})$ samples are needed to observe a collision for the uniform distribution, their algorithm has worst-case sample complexity of $O(\sqrt{n}/\eps)$. They also showed a matching lower bound of $\Omega(\sqrt{n}/\eps)$ samples for any $(\eps,1/3)$-estimate of the sum of the weights. While this sample complexity bound is optimal in the worst case, it might not capture the complexity of the problem on all input instances. In this work, we set out to design algorithms which have better sample complexity bounds on an instance to instance basis. Consider, for example, instances which have a small number of very large weights and the remaining weights are small. With high probability, one would observe a collision in the sampled set with much less number of samples for this instance. Hence, if detecting a collision is the key for an algorithm of this problem to work, then it might require considerably less number of samples compared to the worst-case bound of $\Omega(\sqrt{n}/\eps)$. This leads to our quest for designing an instance optimal algorithm for the sum estimation problem using weighted sampling. Let $p$ denote the proportional (weighted) sampling distribution for $A$. It is known that the probability of observing a collision in the sample is related to the $\ell_2$ norm $\norm{p}_2$ of the sampling distribution $p$ \cite{BC2017,WG2022,DGPP2019}. We ask the following question in this work.

\begin{tcolorbox}[colback=black!5!white,colframe=cyan!50!black,title=Sample complexity for Sum Estimation] Can the sample complexity of the sum estimation problem using proportional sampling be expressed using the $\ell_2$ norm of the proportional sampling distribution $p$? \end{tcolorbox}

We answer this question in the affirmative by designing an algorithm for the sum estimation problem, where the sample complexity depends on the $\ell_2$ norm of $p$. We design an algorithm that produces an $(\eps,\delta)$-estimate of $W$ using $O((\frac{1}{\eps^2 \norm{p}_2}+\frac{\ln 1/\norm{p}_2}{\norm{p}_2})\log 1/\delta)$ samples on an input instance with weighted distribution $p$. More specifically, we prove the following result.

\begin{thm}(Upper bound for Sum Estimation) Given proportional sampling access to a set $A$ with $n$ weighted elements and parameters $0<\eps,\delta< 1$, there exists an algorithm that gives an $(\eps,\delta)$-estimate of the sum of the weights using $O((\frac{1}{\eps^2 \norm{p}_2}+\frac{\ln 1/\norm{p}_2}{\norm{p}_2})\log 1/\delta)$ proportional samples. \end{thm}

Next, we show an almost instance optimal lower bound for the sum estimation problem. More specifically, we prove a $\Omega(\frac{1}{\eps \norm{p}_2})$ lower bound for any $(\eps,1/3)$-estimate of $W$, when the input instance $A$ and $\eps$ are restricted in the following sense.

\begin{thm} Let $A=A_1 \uplus A_2 \uplus \ldots \uplus A_t$, where all elements in $A_i$ have weight $w_i$, and let $t$ be a constant. Let $0<\eps\leq 1$ be such that $\eps\geq 1/\sqrt{|A_i|}$ for all $i$. Then, any $(\eps,1/3)$-estimate of $W$ requires $\Omega\left(\frac{1}{\eps \norm{p}_2}\right)$ proportional samples. \end{thm}

We also construct families of instances for which $\norm{p}_2$ is a constant, and show that for $\eps>1/\sqrt{n}$, one requires $\Omega(1/\eps^2)$ samples for any $(\eps,1/3)$-estimate of the sum. More formally, we prove the following result.

\begin{thm}\label{thm:large-norm} For $\eps>1/\sqrt{n}$, there exist instances with $\norm{p}_2$ being a constant such that any $(\eps,1/3)$ estimate of the sum requires $\Omega(\frac{1}{\eps^2 \norm{p}_2})$ samples. \end{thm}

Next, we discuss our results for the moment estimation problem. We first formally define the $(\eps,\delta)$-moment estimation problem below.

\begin{defn}($(\eps,\delta)$-moment estimation) Given sample access to a set $A$ with $n$ weighted elements, and parameters $\eps,\delta\in (0,1)$, an integer $h>1$, the objective of the $(\eps,\delta)$-moment estimation problem is to obtain an $(\eps,\delta)$-estimate of $W_h(A)=\sum_{a\in A} w(a)^h$. \end{defn}

Motivated by our result for the sum estimation problem for which the sample complexity can be expressed as a function of the $\ell_2$ norm of sampling distribution $p$, we ask whether a similar result can be obtained for the moment estimation problem. We state the following question that we explore in our work.

\begin{tcolorbox}[colback=black!5!white,colframe=brown!75!black,title=Sample complexity for Moment Estimation] Can the sample complexity for the moment estimation problem using proportional sampling be expressed using the $\ell_h$ norm of the proportional sampling distribution $p$? \end{tcolorbox}

We give a positive answer to this question and state the algorithmic result we obtain for this problem below.
\begin{thm}(Upper bound for moment estimation) There exists an algorithm that given proportional sampling access to a set $A$ with $n$ weighted elements and parameters $\eps,\delta\in (0,1)$ and $h>1$, returns an $(\eps,\delta)$-estimate of $W_h(A)$ using $O((\frac{h}{\norm{p}_h}+\frac{1}{\eps^2 \norm{p}_h}+\frac{1}{\eps^2 \norm{p}_2}+\frac{\ln 1/\norm{p}_2}{\norm{p}_2})\log 1/\delta)$ samples. \end{thm}

\section{Technical overview}

The main idea behind our sum estimation algorithm is to find a subset $S\subset A$ for which $w(S)/W$ is sufficiently large such that the ratio $w(S)/W$ can be subsequently estimated using a small number of independent samples to obtain an estimate $\hat{W}$ for $W$. We show that this idea almost works out when $|S|=O(1/\norm{p}_2)$. We first obtain an estimate $T$ for the $\ell_2$ norm $\norm{p}_2$ of the weighted sampling distribution $p$, and then independently sample $S$ of $O(1/T)$ elements. We still need to take care of two issues. First, if the variance of the weights $w_i$ is large, then a sample $S$ of size $O(1/T)$ might not result in a large $w(S)/W$. To overcome this, we proceed as follows. First, we check whether there exists a {\em heavy hitter} element $a_i\in A$ for which the sampling probability $p_i$ is larger than some threshold. If such an element $a_i$ exists, we set $S=\{a_i\}$. Else, the variance of the weights cannot be too large and we show that a sampled set $S$ of size $O(1/T)$ would have large $w(S)/W$ with high probability. Second, we estimate the ratio $w(S)/W$ by obtaining independent samples and computing the fraction of sampled elements that belong to $S$. One subtle issue here is that the estimate will be wrong if the sample $S$ contains duplicates. We resolve this by choosing the sample size $O(1/T)$ such that with high probability $S$ does not contain any repeated elements.

For the moment estimation problem, Bhattacharya and Pradhan \cite{BP2025} have given an $(\eps,\delta)$-estimate for the $h$-th moment $W_h(A)$ of $A$ using $O((\sqrt{n}/\eps+\frac{n^{1-1/h}}{\eps^2})\ln 1/\delta)$ proportional samples, for any $h>1$. We note that the first term in the above sample complexity bound appears due to the sum estimation algorithm of \cite{BT2024} that the authors use as a subroutine. We also observe that the second term can be expressed using the $\ell_h$ norm $\norm{p}_h$ of the sampling distribution $p$. In our algorithm for the moment estimation problem, we first estimate the $\ell_h$ norm $\norm{p}_h$ of $p$. We do this by sampling elements until we observe the first $h$-wise collision. We show that this gives a good estimate of $\norm{p}_h$, and we run the algorithm of Bhattacharya and Pradhan with this estimate. We use our sum estimation algorithm as a subroutine in this algorithm in place of the sum estimation algorithm of \cite{BT2024}.

We prove a lower bound on the number of samples required to obtain an $(\eps,1/3)$-estimate the sum of the weights of $A$ using proportional sampling. We adapt the techniques of Beretta and T{\v{e}}tek \cite{BT2024} to obtain an almost instance optimal lower bound for the sum estimation problem. However, we have to face significant difficulties in extending the worst-case lower bound of \cite{BT2024} to an almost instance optimal lower bound. First, the construction of the hard instances differs substantially from that of \cite{BT2024}. Beretta and T{\v{e}}tek \cite{BT2024} showed that an instance with uniform weights is a worst-case instance for the sum estimation problem. Since, we aim for an instance optimal lower bound, for every input instance $I$, we need to construct an almost similar instance $I'$ with $W(I')\geq (1+\eps)W(I)$ such that it is difficult to distinguish the two instances while using a small number of samples. Beretta and T{\v{e}}tek computed the ratio of the conditional probabilities of observing a sample from $I$ and $I'$, and show that the ratio is a function of the number of distinct elements in the sample. Using concentration bounds on the number of distinct elements in the sample, they argue that this ratio remains close to one when the number of samples is $o(\sqrt{n}/\eps)$, and show that even the Bayes optimal classifier cannot distinguish the instances when this ratio is close to one. We group elements of $A$ into $t$ blocks $A_1,\ldots,A_t$, where all elements in a block have the same weight. We compute the ratio of the conditional probabilities of observing a sample from $I$ and $I'$, and show that this ratio depends on the number of distinct elements observed from each block. We show that the number of distinct elements from each block is concentrated around its expected value and argue that the ratio remains close to one when the number of samples is $o(\frac{1}{\eps \norm{p}_2})$. Since the ratio remains close to one, even the Bayes optimal classifier cannot distinguish whether the samples are from $I$ or $I'$ with probability at least $2/3$. We require the number $t$ of blocks to be a constant and $\eps\geq 1/\sqrt{|A_i|}$ for all $i$ for this analysis to work.

\section{Related works} Motwani~\etal~\cite{MPX2007} initiated the study of the sum estimation problem using weighted sampling and designed an algorithm that uses $\tilde{O}(\sqrt{n}/\eps^{7/2})$ proportional samples. Recently, Beretta and T{\v{e}}tek \cite{BT2024} designed an worst-case optimal algorithm for the $(\eps,1/3)$-sum estimation problem that uses $\Theta(\sqrt{n}/\eps)$ samples. For the moment estimation problem, Bhattacharya and Pradhan \cite{BP2025} designed an algorithm for $(\eps,\delta)$-moment estimation using $O(\frac{n^{1-1/h}\ln 1/\delta}{\eps^2})$ samples for any $h>1$. There have been a lot of work on graph parameter estimation problems in the query model. We cite only a few relevant works here \cite{ERS2018,ERS2019,BCS2026}.

\section{Preliminaries}

\begin{fact}\label{fact:norms} For any vector $x\in \C^n$, and for any $0<r<p$, we have $||x||_p\leq ||x||_r\leq n^{(1/r-1/p)} ||x||_p$. \end{fact}

\begin{lem}(Chernoff bound \cite{MU2017})\label{lem:chernoff} Let $Z_1,Z_2,\ldots,Z_v$ be independent and identically distributed Bernoulli random variables and $0<\gamma<1$ and $\eta>0$. Let $Z=\sum_{i=1}^v Z_i$. Then, $\Pr[Z\leq(1-\gamma) \EE[Z]]\leq e^{-1/2 \cdot \gamma^2 \cdot \EE[Z]}$ and $\Pr[Z\geq (1+\eta)\EE[Z]]\leq \left(\frac{e^{\eta}}{(1+\eta)^{(1+\eta)}}\right)^{\EE[Z]}$.Also, for $R\geq 2e \EE[Z]$, we have $\Pr[Z\geq R]\leq 2^{-R}$. \end{lem}

\begin{thm}(Efron-Stein Inequality \cite{BLM2013})\label{thm:es} Let $X_1,\ldots,X_n$ be independent random variables and let $X'_1,\ldots,X'_n$ be an independent copy of $X_1,\ldots,X_n$. Let $Z=f(X_1,\ldots,X_n)$ be a square-integrable function of $X$. Let us define $Z'_i=f(X_1,\ldots,X_{i-1},X'_i,X_{i+1},\ldots,X_n)$ for all $i$. Then, $\var[Z]\leq 1/2 \cdot \sum_{i=1}^n \EE[(Z-Z'_i)^2]$.\end{thm}

Suppose we have an $(\eps,1/3)$-estimator for $S_t(A)$ which uses $T$ samples. There is a standard technique called the \emph{median of means} (see Theorem 2.1 of \cite{AMS1999}) to reduce the failure probability from $1/3$ to any desired value $\delta \in (0,1)$, using $O(T \cdot \ln \frac{1}{\delta})$ samples. For the sake of completeness, we present the technique and its proof below.

\begin{lem}(Median of means)\label{lem:mom} Let $Y_1, Y_2, \ldots, Y_v$ denote the outputs of $v$ independent runs of an $(\eps, 1/3)$-estimator for $S_t(A)$, where $v = 48 \cdot \ln(1/\delta)$. Let $Y_{med}$ denote the median of $Y_1,\ldots,Y_v$. Then, $Y_{med}$ is an $(\eps,\delta)$-estimate of $S_t(A)$. \end{lem}

\begin{proof} Since $Y_i$ is an $(\eps,1/3)$-estimate of $S_t(A)$, with probability at least $2/3$, we have $Y_i \in (1 \pm \eps) S_t(A)$. Let $Z_i = 1$ if $Y_i \in (1 \pm \eps) S_t(A)$, and $Z_i = 0$ otherwise. Let $Z = \sum_{i=1}^v Z_i$. Then, $\EE[Z] \geq (2/3) \cdot v$. If $Y_{med} \notin (1 \pm \eps) S_t(A)$, at most $v/2$ of the $Y_i$'s belong to $(1 \pm \eps) S_t(A)$, i.e., $Z \leq v/2$. Using Lemma \ref{lem:chernoff}, we get $\Pr[Y_{med} \notin (1 \pm \eps) S_t(A)] \leq \Pr[Z \leq v/2] \leq e^{-v/48}$. For $v = 48 \cdot \ln(1/\delta)$, this probability is at most $\delta$.
\end{proof}

Below, we state two lemma from Beretta and T{\v{e}}tek \cite{BT2024} that we use in our lower bounds.

\begin{lem}(Lemma 4.3 of \cite{BT2024})\label{lem:43BT2024} Let $\eps<1/3$. Let $N=n$ with probability $1/2$ and $N=(1-\eps)n$ otherwise. Consider a random instance of the sum estimation problem $(U,w)$ with $|U|=N$ and $w(a)=1$ for $a\in U$. Consider a sample of size $m$ and its fingerprint $F_m=\{(c_i,w(a_i))\}_{i=1,\ldots,l}$ and define the ratio $\RR(F_m)=\frac{\Pr[F_m|N=n]}{\Pr[F_m|N=(1-\eps)n]}$. If $m=o(\sqrt{n}/\eps)$ and $m=o(n)$, then $\Pr[98/100\leq \RR(F_m)\leq 100/98]\geq 99/100$ for large enough $n$. \end{lem}

\begin{lem}(Lemma 4.4 of \cite{BT2024})\label{lem:44BT2024} Let $0<\eps<p$ and set $q=p$ with probability $1/2$ and $q=p-\eps$, otherwise. Let $X_1,\ldots,X_m$ be a sequence of iid Bernoulli random variables with bias $q$, and let $l=\sum_{i=1}^m X_i\sim Bin(m,q)$. Define the ratio $\RR(F_m)=\frac{\Pr[l|q=p]}{\Pr[l|q=p-\eps]}$. Then, for $m=o(p/\eps^2)$, $\Pr[98/100\leq \RR(l)\leq 100/98]\geq 99/100$, for $p$ and thus $\eps$ small enough. \end{lem}


\section{Algorithm for Sum Estimation}\label{sec:alg-sum-est} 

In this section, we describe an algorithm for computing an $(\eps,\delta)$-estimate of the sum $W$ of the weights of elements in $A$ and analyze the sample complexity of the algorithm. The algorithm has three phases. In the first phase, the algorithm samples elements using proportional sampling until it observes the first {\em collision}\footnote{We say a collision happens when the same element is sampled twice.}. Let $Z$ denote the set of sampled elements until it observes the first collision. The algorithm obtains an estimate $T=1/|Z|$ for the $\norm{p}_2$ of the sampling distribution $p$. Next, it checks if there exists an element $a_i\in A$ for which $p_i$ is sufficiently large, say some constant times $\norm{p}_2$ and if such an element exists, sets $S=\{a_i\}$. If there does not exist such an element in $A$, it independently samples $1/3000T$ elements following the distribution $p$. Let $S$ denote the set of sampled elements during this phase. Let $w(S) = \sum_{a \in S} w(a)$ denote the total weight of the elements in $S$. In the third phase, the algorithm obtains new samples from $A$ to estimate $\hat{p}$ of $w(S)/W$. The algorithm computes $\widehat{W}=w(S)/\hat{p}$ as an estimate of $W$. To bound the sample complexity of the algorithm, we analyze the number of samples needed in each phase. First, we bound the number of samples required to observe the first collision. Next, we show that $w(S)/W$ is sufficiently large, and hence, we can estimate $\frac{w(S)}{W}$ using a small number of samples and use it to estimate $W$.

Let $p$ denote the proportional sampling distribution for $A$, and let $\norm{p}_2$ denote the $l_2$ norm of $p$. We show that the number of samples required to observe a collision is closely related to $\norm{p}_2$. In particular, we show that $O(1/\norm{p}_2)$ samples are sufficient to observe at least one collision with constant probability. Moreover, for the sample $S$ generated as above, the quantity $\frac{w(S)}{W}$ can be shown to be at least some constant times $\norm{p}_2$ with probability at least $8/9$. Therefore, using $O(\frac{1}{\eps^2 \norm{p}_2})$ samples, we can obtain a good estimate of $\frac{w(S)}{W}$ with constant probability. The sample complexity of the algorithm is at most $O((\frac{1}{\eps^2 \norm{p}_2}+\frac{\log 1/\norm{p}_2}{\norm{p}_2})\log 1/\delta)$. Next, we state the main claim for Algorithm \ref{alg:estsum1} and give a formal analysis of its sample complexity.

\begin{algorithm}[htbp!]
    \caption{Sum Estimation using Proportional Sampling }
    \label{alg:estsum1}
    \begin{algorithmic}[1]
        \Procedure{SumEstimator}{$A, \varepsilon, \delta$}
            \State $r \gets 48 \log(1/\delta)$
            \For{$s = 1$ to $r$}

                \State $Z \gets \emptyset$ \Comment Estimate $\ell_2$ norm of $p$ and assign its value to $T$
                \While{true}
                    \State Sample element $a$ following distribution $p$
                    \If{$a \notin Z$}
                        \State Add $a$ to $Z$
                    \Else
                        \State \textbf{break}
                    \EndIf
               \EndWhile
		\State $T = 1/|Z|$ \Comment Estimate $\ell_2$ of $p$ and assign its value to $T$\newline

		\State Let $q=\frac{T}{4\cdot 3000\cdot 576^2\cdot 45^2}$. Let $U$ be a sample of size $10/q$ following distribution $p$. 
		\State Let $V$ denote a multiset of $192/q \cdot \log 1/q$ independent samples from $p$.
		\State Check if there exists an element $a_i$ in $U$ that appears in $V$ with empirical frequency at least $q-\gamma/2$, where $\gamma=q/2$.
		\If{Yes}
			\State Set $S=\{a_i\}$. \Comment Algorithm detects a heavy hitter
		\Else
			\State $S \gets \emptyset$.\Comment Algorithm samples elements such that $w(S)/W$ is large.
                	\For{$j=1$ to $\frac{1}{3000T}$}
                    		\State Sample an element $a$ following distribution $p$
				\State Add $a$ to $S$
                	\EndFor \newline
		\EndIf 

                \State $l \gets \frac{24 \ln(120)}{\eps^2 \cdot q}$, $h \gets 0$ \Comment Estimate $w(S)/W$
		\For{$i=1$ to $l$}
                    \State Sample an element $a$ following distribution $p$
                    \If{$a \in S$}
                        \State $h \gets h + 1$
                    \EndIf
                \EndFor								\Comment Estimate $w(S)/W$

                \State $w(S) \gets \sum_{a \in S} w(a)$
                \State $\widehat{p} \gets \frac{h}{l}$
                \State $\widehat{W}_s \gets \frac{w(S)}{\widehat{p}}$
            \EndFor
            \State \Return $\widehat{W}\gets \text{median}(\widehat{W}_1, \widehat{W}_2, \ldots, \widehat{W}_r)$
        \EndProcedure
    \end{algorithmic}
\end{algorithm}

\begin{claim} Let $0<\eps,\delta< 1$. Algorithm \ref{alg:estsum1} gives an $(\eps,\delta)$-estimate $\widehat{W}$ of $W$ and uses at most $O((\frac{1}{\eps^2 \norm{p}_2}+\frac{\log 1/\norm{p}_2}{\norm{p}_2})\log 1/\delta)$ proportional samples. \end{claim}

\begin{proof} We bound the number $|Z|$ of samples required to observe the first collision with constant probability. In particular, for $|Z| = \Theta(\frac{1}{\norm{p}_2})$, we show that the first while loop terminates with probability at least $7/8$. We prove the following result.

\begin{lem}\label{lem:range-wst-T} We have $\Pr[|Z|\leq \frac{1}{24 \norm{p}_2}]\leq 1/16$ and $\Pr[|Z|> \frac{15}{\norm{p}_2}] \leq \frac{1}{16}$. In particular, $\Pr[1/24 \cdot 1/\norm{p}_2<|Z|\leq \frac{15}{\norm{p}_2}]\geq 7/8$. \end{lem}

\begin{proof} Using a result of Schulte-Geers and Waggoner \cite{WG2022}, we obtain, for any $\eta \in (0,1)$, $\Pr[|Z|\leq \max\{e^2,2\ln(\frac{1}{\eta})\}\frac{2}{\norm{p}_2}]\geq 1-\eta$ and $\Pr[|Z|\leq (\frac{\eta}{e})\frac{2}{\norm{p}_2}]\leq \eta$. For $\eta=1/16$, we obtain the above result. \end{proof}

For the following, we condition on the event $\E_1$ that Lemma \ref{lem:range-wst-T} holds. Let $T=1/|Z|$ be an estimate of $\norm{p}_2$. Conditioned on $\E_1$, we have $\norm{p}_2/15<T<24\norm{p}_2$. 

Let $q=\frac{T}{4\cdot 3000\cdot 576^2\cdot 45^2}$. Next, the algorithm decides if there exists a {\em heavy hitter} element $a_i$ in the instance with $p_i\geq q$ or for all elements $a_i$, we have $p_i< q-\gamma=q/2$, where we set $\gamma=q/2$. Below, we argue that the algorithm with probability at least $0.98$ decides correctly whether there exists such a {\em heavy hitter} element in the instance. 

\begin{claim} Algorithm \ref{alg:estsum1} decides correctly whether there exists an element $a_i$ with $p_i\geq q$ or $p_i\leq q/2$ for all $a_i\in A$, with probability at least $0.98$. \end{claim}

\begin{proof} Suppose there exists an element $a_i\in A$ with $p_i\geq q$. This element is sampled with probability at least $q$, and will be included in $U$ with probability at least $1-e^{-10}\geq 0.99$. Using a Chernoff bound argument, using a sample size of at least $192/q \ln 1/q$, this {\em heavy hitter} element with probability at least $(1-q^4)$ will have empirical frequency in $V$ at least $q-\gamma/2$. Therefore, Algorithm \ref{alg:estsum1} will output Yes with probability at least $0.98$. 

Suppose for any element $a_i$ in $A$, we have $p_i<q-\gamma=q/2$. We need to argue that with high probability, none of these elements appear with significantly higher empirical frequency in $V$. Let us consider any element $a_i$ in $U$ with probability $p_i<q/2$. Using a sample size of $192/q \ln 1/q$, and using Chernoff bounds, the probability that its empirical frequency in $V$ is at least $q-\gamma/2$ is at most $q^{8}$. Since there are at most $10/q$ elements in $U$, the probability that Algorithm \ref{alg:estsum1} identifies any such element as a heavy hitter is at most $0.01$.

\end{proof}

Let us condition on event $\E_2$ that Algorithm \ref{alg:estsum1} correctly decides if there exists $a_i$ such that $p_{\max}=\max_{i=1}^n p_i \geq q$ or for any $a_i\in A$, we have $p_{\max}<q/2$. We have $\Pr[\E_2]\geq 0.98$. If Algorithm \ref{alg:estsum1} finds a heavy hitter element $a_i$, it sets $S=\{a_i\}$. Since the algorithm might return anything if $q/2<p_{\max}=\max_{i=1}^n p_i<q$, the heavy hitter element $a_i$ detected by the algorithm will have $p_i\geq q/2$. If there does not exist any {\em heavy hitters} in the instance, the algorithm samples $1/3000T$ independent samples following distribution $p$. Let $S$ denote the set of sampled elements. Since $|S|=1/3000T$, we have $1/72000\leq |S|\cdot \norm{p}_2\leq 1/200$.

\begin{claim} With probability at least $0.99$, there are no repeated elements in $S$. \end{claim} 
\begin{proof} Using Lemma \ref{lem:range-wst-T}, for $\eta=0.01$, with probability at least $0.99$, the number of samples required for a collision is at least $1/50e \cdot 1/\norm{p}_2$. Since we have $\frac{1}{15T}<1/\norm{p}_2$, the number of samples in $S$ is at most $1/3000T<1/200 \cdot 1/\norm{p}_2\leq 1/50e \cdot 1/\norm{p}_2$. Hence, with probability at least $0.99$, there is no repeated element in $S$. \end{proof}

Let $\E_3$ denote the event that there are no repeated elements in $S$. For the following, we condition on $\E_3$.

\begin{lem}\label{lem:relate-wst-w} We have $\Pr[\frac{w(S)}{W} \geq q/2]>8/9$. \end{lem}

\begin{proof} If Algorithm \ref{alg:estsum1} detects a heavy hitter element $a_i$ in $A$, then we have $w(S)/W=p_i\geq q/2$. If Algorithm \ref{alg:estsum1} did not detect any heavy hitters, we proceed as follows. We can write $\frac{w(S)}{W}=\sum_{a_i\in S} \frac{w(a_i)}{W}= \sum_{a_i\in S} p_i$. We have $p_{\max}=\max_i p_i \leq q$. Let $X$ be any sample following distribution $p$. We have $\EE[X]=W\cdot \norm{p}_2^2$. Hence, $\EE[w(S)/W]=|S|\cdot \norm{p}_2^2$ and $\var[w(S)/W]\leq |S|\cdot \norm{p}_3^3$. Since $p_{\max}\leq q$, we have $\norm{p}_3^3\leq q \cdot \norm{p}_2^2$. Using this we can write, $\var[w(S)/W]\leq |S| \cdot q \cdot \norm{p}_2^2\leq 1/3000T \cdot q\cdot \norm{p}_2^2$. The standard deviation of $w(S)/W$ will be at most $\sigma(w(S)/W)\leq \sqrt{\frac{q}{3000T}} \norm{p}_2$. Using Chebyshev's inequality, with probability at least $8/9$, $w(S)/W\geq \EE[w(S)/W]-3\cdot \sigma(w(S)/W)\geq 1/3000T \cdot T^2/576 - 3\cdot \sqrt{\frac{q}{3000T}}\cdot \norm{p}_2$. For $q=\frac{T}{4\cdot 3000\cdot 45^2\cdot 576^2}$, we will get with probability at least $8/9$, $w(S)/W\geq \frac{T}{2\cdot 3000\cdot 576}$. Hence, with probability at least $8/9$, $w(S)/W\geq \min\{q,\frac{T}{2\cdot 3000\cdot 576}\}=q$.
\end{proof}

We have shown that $w(S)/W$ is at least $q/2$ with probability at least $8/9$. We condition on the event $\E_4$ that $w(S)/W\geq q/2$. We will show that $\wpp=h/l$ provides an $(\eps/2,1/60)$-estimate of $w(S)/W$ for $l = O(\frac{1}{\eps^2 \norm{p}_2})$.

\begin{lem}\label{lem:est-wst-w} $h/l$ is an $(\eps/2, 1/60)$ estimate of $w(S)/W$ for $l= O(\frac{1}{\eps^2 \norm{p}_2})$. \end{lem}

\begin{proof} We obtain independent samples $Y_1,Y_2,\ldots,Y_l$ from $A$. Now, we define a random variable

\[h_i=\begin{cases}
	1 & Y_i\in S\\0 & otherwise
\end{cases}\]

Let $h=\sum_{i=1}^l h_i$ denote the number of sampled elements that belong to the set $S$. We can write $\EE[h] = \sum_{i=1}^l \EE[h_i] = \sum_{i=1}^l \frac{w(S)}{W} = l\cdot \frac{w(S)}{W}$. Hence, $1/l \cdot \EE[h]=w(S)/W$. Using Chernoff bound, we have $\Pr[|h-\EE[h]|>\eps/2 \cdot \EE[h]]\leq 2e^{-\frac{1/4 \cdot \EE[h]\eps^2}{3}}$. We have $\Pr[|h/l-\EE[h/l]|>\eps/2 \cdot \EE[h/l]]\leq 2e^{-\frac{\EE[h]\eps^2}{12}}$. We need to choose $l$ such this probability $2e^{-\frac{l\cdot w(S)/W \cdot \eps^2}{12}}=\frac{1}{60}$. Conditioned on $\E_4$, this implies $l\geq \frac{24 \ln(120)}{\eps^2 q}$. Hence, $h/l$ gives an $(\eps/2,1/60)$ estimate of $w(S)/W$ and the required sample size is $l=O\left(\frac{1}{\eps^2 \norm{p}_2}\right)$.
\end{proof}

We have shown that $\hat{p}=h/l$ is an $(\eps/2,1/60)$ estimator of $w(S)/W$. Algorithm \ref{alg:estsum1} in iteration $s$ returns $\widehat{W_s}=w(S)/\hat{p}$ as an estimate of $W$. Next, we show that $\widehat{W_s}$ is an $(\eps,1/3)$-estimate of $W$. First, we show that $(1-\eps)W\leq \widehat{W_s}\leq (1+\eps)W$, and then we bound the probability that $\widehat{W_s}$ falls outside this range. Using Lemma \ref{lem:est-wst-w}, we have $(1-\eps/2)\frac{w(S)}{W} \leq \frac{h}{l}\leq (1+\eps/2)\frac{w(S)}{W}$. Rearranging the terms and using $\hat{p}=h/l$ and $\eps\leq 1$, we obtain $(1-\eps)W \leq \widehat{W_s}\leq (1+\eps)W$. Next, we bound the failure probability of the event that $\widehat{W_s} \notin (1 \pm \eps)W$. Algorithm \ref{alg:estsum1} fails if either $\E_1$ or $\E_2$ or $\E_3$ or $\E_4$ does not hold or if $\hat{p}$ is not an $(\eps/2,1/60)$-estimator of $w(S)/W$. The total failure probability can be bounded as $1/8+1/50+1/100+1/9+1/60\leq 1/3$. Once we have an $(\eps, 1/3)$-estimate of $W$, we take $r=48\log(1/\delta)$ independent copies of this estimator and use Lemma \ref{lem:mom} to obtain an $(\eps,\delta)$-estimate of $W$.

Next, we analyse the sample complexity of our algorithm. The number of samples required to observe the first collision is $\Theta(1/\norm{p}_2)$. Checking whether a heavy hitter element exists requires $O(1/q \ln 1/q)=O(\frac{\ln 1/\norm{p}_2}{\norm{p}_2})$ samples. Estimating $w(S)/W$ requires at most $O(\frac{1}{\eps^2 \norm{p}_2})$ samples. Repeating this for at most $48\log 1/\delta$ times, overall sample complexity is $O((\frac{1}{\eps^2 \norm{p}_2}+\frac{\ln 1/\norm{p}_2}{\norm{p}_2})\log 1/\delta)$.
\end{proof}

\section{Algorithm for Moment Estimation} In this section, we design an algorithm for an $(\eps,\delta)$-moment estimation using weighted sampling. We start by describing the high-level idea of our algorithm. Bhattacharya and Pradhan \cite{BP2025} have given an algorithm for an $(\eps,1/3)$-moment estimation that for $h>1$ uses $O(\sqrt{n}/\eps+n^{1-1/h}/\eps^2)$ samples. In this work, we design algorithms for an $(\eps,\delta)$ moment estimation with better sample complexity bounds on an instance to instance basis. First, we observe that the sample complexity of the $(\eps,1/3)$-moment estimation algorithm of Bhattacharya and Pradhan \cite{BP2025} can be written as $O(\sqrt{n}/\eps+\frac{1}{\eps^2 \norm{p}_h})$. However, since the value of $\norm{p}_h$ is not known for the proportional sampling distribution $p$, we cannot directly use their algorithm. One natural idea is to first estimate $\norm{p}_h$ of the distribution $p$, and then run the algorithm of \cite{BP2025} with this estimate. In order to accomplish this, we observe that $\norm{p}_h^h$ captures the probability of a $h$-way collision for distribution $p$. Our algorithm keeps on sampling until it observes the first $h$-way collision and we show that the expected number of samples required to observe a $h$-way collision can be written as $\Theta(h/\norm{p}_h)$. Next, we observe that the $O(\sqrt{n}/\eps)$ term in the sample complexity for moment estimation by Bhattacharya and Pradhan appears due to the sum estimation algorithm of \cite{BT2024} that the authors use as a subroutine in their algorithm. We use Algorithm \ref{alg:estsum1} to estimate the sum of the weights. The overall algorithm as motivated above can be described as follows. In the first phase, we keep on sampling until we observe a $h$-way collision. Suppose this requires $L$ samples. In the second phase, we obtain independent $O(L/(h\eps^2))$ samples and run the algorithm of Bhattacharya and Pradhan \cite{BP2025} with these samples as input. We will show that Algorithm \ref{alg:estmom} gives an $(\epsilon,1/3)$ estimate of $W_h(A)$ for $h>1$ and the sample complexity of Algorithm \ref{alg:estmom} is $O(\frac{h}{\norm{p}_h}+\frac{1}{\eps^2 \norm{p}_h}+\frac{1}{\eps^2 \norm{p}_2}+\frac{\ln 1/\norm{p}_2}{\norm{p}_2})$, where the first phase requires $O(\frac{h}{\norm{p}_h})$ samples to observe a $h$-way collision, and the second phase requires $O(\frac{1}{\eps^2 \norm{p}_h})$ samples to estimate $W_h(A)$. Next, we present the algorithm and argue its correctness.

\begin{algorithm}[htbp!]
\caption{Moment Estimation using Proportional Sampling }
\label{alg:estmom}
\begin{algorithmic}[1]
	\Procedure{MomentEstimator}{$A,h,\eps$}
        	\State $\eps_1\gets \eps/2$
                \State $C[i]\gets 0\; \forall i\in [n]$
                \While{true}
                	\State Sample an element $a_i$ following distribution $p$
                    	\State $C[i]\gets C[i]+1$
                    	\If{$C[i]$ equals $h$}
                        	\State \textbf{break}
                    	\EndIf
                \EndWhile
                \State $L \gets \sum_{i=1}^n C[i]$
                \State $\widetilde{W} \gets $ \hyperlink{alg:estsum1}{\textsc{SumEstimator}}$(A,\epsilon_1,1/8)$
                \For{$j=1 \text{ to } m=\frac{30720\cdot L}{h\eps^2}$}
                	\State Sample an element $b_j$ from $A$ following distribution $p$
	        	\State $\tilde{p_j}\gets \frac{w(b_j)}{\widetilde{W}}$
                	\State $X_j\gets \frac{w(b_j)^h}{\tilde{p_j}}$
                \EndFor
                \State\Return $X=\frac{1}{m}\sum\limits_{j=1}^m X_j$
        \EndProcedure
    \end{algorithmic}
\end{algorithm}

\begin{claim} Let $h>1$. Algorithm \ref{alg:estmom} gives an $(\eps,1/3)$-estimate of $W_h(A)$ using $O(\frac{h}{\norm{p}_h}+\frac{1}{\eps^2 \norm{p}_h}+\frac{1}{\eps^2 \norm{p}_2}+\frac{\ln 1/\norm{p}_2}{\norm{p}_2})$ samples. \end{claim}

\begin{proof} Our algorithm has two phases. In the first phase, the algorithm keeps on sampling elements until it observes the first $h$-way collision\footnote{A $h$-way collision happens when an element is sampled $h$ times.}. Let $L$ denote the number of samples required to observe the first $h$-way collision. Then, using a result of Schulte-Geers and Waggoner \cite{WG2022}, we get for any $\eta \in (0,1)$, $\Pr[L\leq \max\{e^2,2\ln(\frac{1}{\eta})\}\frac{h}{\norm{p}_h}]\geq 1-\eta$ and $\Pr[L\geq (\frac{\eta}{e})\frac{h}{\norm{p}_h}]\geq 1-\eta$. For $\eta=1/16$, this implies that $\Pr\left[\frac{h}{48 \norm{p}_h}\leq L \leq \frac{8h}{\norm{p}_h}\right] \geq \frac{7}{8}$. Let $\E_1$ denote the event that $\frac{h}{48 \norm{p}_h}\leq L \leq \frac{8h}{\norm{p}_h}$, and we have $\Pr[\E_1]\geq 7/8$. Conditioned on $\E_1$, this phase requires at most $\Theta(\frac{h}{\norm{p}_h})$ samples. The algorithm then obtains an $(\eps/2,1/8)$-estimate $\widetilde{W}$ of the sum of the weights of $A$ using Algorithm \ref{alg:estsum1}. 

Next, the algorithm obtains $m=\frac{30720\cdot L}{h\eps^2}=O(\frac{1}{\eps^2\cdot \norm{p}_h})$ independent samples. Let $b_j$ denote the element sampled during iteration $j$, where $b_j$ has weight $w(b_j)$ and we set $\tilde{p_j}=w(b_j)/\wt{W}$ and $X_j=w(b_j)^h/\tilde{p_j}$. We show that $X=1/m \sum_{j=1}^m X_j$ gives an $(\eps,1/3)$-estimate of $W_h(A)$.

Let $\E_2$ denote the event that $\wt{W}\in [(1-\eps_1)W,(1+\eps_1)W]$. Using Algorithm \ref{alg:estsum1} with $\delta=1/8$, we have $\Pr[\E_2]\geq 7/8$. Conditioned on $\E_2$, for any element $b_j\in A$, we have $\frac{p_j}{1+\eps_1}\leq \tilde{p_j}\leq \frac{p_j}{1-\eps_1}$. We can write $\EE[X_j]=\sum_{a\in A} \frac{w(a)^h}{\tilde{p_a}}\cdot p_a$. This gives $(1-\eps_1)W_h(A) \leq \EE[X_j]\leq (1+\eps_1)W_h(A)$. Since $X_1,X_2,\ldots,X_m$ are independent random variables, using linearity of expectation, we have $\EE[X]=\frac{1}{m}\sum_{j=1}^m \EE[X_j]=\EE[X_j]\in [(1-\eps_1)W_h(A), (1+\eps_1)W_h(A)]$. Since $X$ is not an unbiased estimator of $W_h(A)$, we need to bound its variance. 

\begin{claim} For $m=\frac{30720\cdot L}{\eps^2 h}\geq \frac{640}{\eps^2 \norm{p}_h}$ and $\eps_1=\eps/2$, we have $\Pr[|X-W_h(A)|>\eps W_h(A)]\leq 1/40$. \end{claim}
\begin{proof} We start by bounding the variance of $X_j$.

\begin{equation*}
        \begin{aligned}
        \var[X_j]       & \leq \EE[X_j^2]= \sum_{a\in A}\frac{w(a)^{2h}}{\tilde{p_a}^2}\cdot p_a \\
                        & \leq (1+\epsilon_1)^2\sum_{a\in A}\frac{w(a)^{2h}}{p_a} \\
                        & = (1+\epsilon_1)^2\cdot W \cdot \sum_{a\in A}\frac{w(a)^{2h}}{w(a)} \\
                        & = (1+\epsilon_1)^2\cdot W \cdot \sum_{a\in A} w(a)^{2h-1}
        \end{aligned}
\end{equation*}

Since $X_1,\ldots,X_m$ are independent random variables, we have $\var[X]=\frac{\var[X_j]}{m}\leq \frac{(1+\eps_1)^2\cdot W}{m} \cdot \sum\limits_{a\in A} w(a)^{2h-1}$. Using Chebyshev's inequality, we get

\begin{equation*} 
        \begin{aligned}
        \Pr(|X-W_h(A)|>\eps W_h(A))     & \leq \Pr(|X-\EE[X]|>(\eps-\eps_1)W_h(A))\\
                                        & \leq \frac{\var[X]}{(\eps-\eps_1)^2 W_h(A)^2}\\
                                        & \leq \frac{(1+\eps_1)^2\cdot W \cdot \sum_{a\in A} w(a)^{2h-1}}{m(\eps-\eps_1)^2 W_h(A)^2}\\
                                        & = \frac{(1+\eps_1)^2\cdot W \cdot ||w(A)||_{2h-1}^{2h-1}}{m(\eps-\eps_1)^2\cdot(||w(A)||_h^h)^2}\\
                                        & \leq \frac{(1+\eps_1)^2\cdot W \cdot (||w(A)||_{h}^{h})^{2-1/h}}{m(\eps-\eps_1)^2\cdot(||w(A)||_h^h)^2} &\hfill [Using~Fact~\ref{fact:norms},~||w(A)||_{2h-1}\leq ||w(A)||_h]\\
                                        & = \frac{(1+\eps_1)^2\cdot W }{m(\eps-\eps_1)^2\cdot||w(A)||_h}\\
                                        & = \frac{(1+\eps_1)^2}{m(\eps-\eps_1)^2}\cdot \frac{1}{||p||_h}\\
        \end{aligned}
\end{equation*}

For $m=\frac{30720\cdot L}{\eps^2 h}\geq \frac{640}{\eps^2 \norm{p}_h}$ and $\eps_1=\eps/2$, the above probability is at most $1/40$.

\end{proof}

Let $\E_3$ denote the event that the estimate $X$ is not in $[(1-\eps)W_h(A),(1+\eps)W_h(A)]$, and using the above, we get $\Pr[\E_3]\leq 1/40$. Our algorithm fails if either $\E_1$ fails or $\E_2$ fails or $\E_3$ becomes true. Overall failure probability is bounded by $1/8 + 1/8 + 1/40\leq 1/3$.

Next, we discuss the sample complexity of Algorithm \ref{alg:estmom}. Our algorithm requires at most $O(\frac{h}{\norm{p}_h})$ samples in the first phase. We used Algorithm \ref{alg:estsum1} as a subroutine to obtain estimate $\wt{W}$ of the sum of weights, and for this, Algorithm \ref{alg:estsum1} requires at most $O(\frac{1}{\eps^2\norm{p}_2}+\frac{\ln 1/\norm{p}_2}{\norm{p}_2})$ samples. In the final phase, we need $O(\frac{1}{\eps^2 \norm{p}_h})$ samples to estimate the moment $W_h(A)$. Hence, Algorithm \ref{alg:estmom} gives an $(\eps,1/3)$ estimate of $W_h(A)$ using $O(\frac{h}{\norm{p}_h}+\frac{1}{\eps^2 \norm{p}_2}+\frac{\ln 1/\norm{p}_2}{\norm{p}_2}+\frac{1}{\eps^2 \norm{p}_h})$ samples.
\end{proof}

Using Lemma \ref{lem:mom}, one can convert the above $(\eps,1/3)$-estimator into an $(\eps,\delta)$-estimator using $O((\frac{h}{\norm{p}_h}+\frac{1}{\eps^2 \norm{p}_2}+\frac{\ln 1/\norm{p}_2}{\norm{p}_2}+\frac{1}{\eps^2 \norm{p}_h})\log 1/\delta)$ samples.

\section{Lower Bound for Sum Estimation} \label{sec:lb} 

In this section, we prove a lower bound on the number of samples required to obtain an $(\eps,1/3)$-estimate of the sum of the weights of the elements of $A$ using proportional sampling. We start by describing the high level idea of the proof. For every input instance $I(A,W)$, where $A$ denotes the set of input elements and $W$ denotes the sum of the weights of the elements in $A$, we construct another instance $I'(A',W')$ such that the sums of the weights of the elements of the two instances differ by at least a factor of $(1\pm \eps)$. We will show that, given a small number of samples, it is not possible to decide whether the samples come from $I$ or $I'$ with probability at least $2/3$. We assume that an adversary chooses instance $I$ or $I'$ with equal probability, and generates $m$ proportional samples from the chosen instance. Let $S$ denote the (multi) set of $m$ sampled elements. The fingerprint $F_m$ of a sample $S$ is defined as a collection of tuples of the form $(c_a,w(a))$, where for each distinct element $a\in S$, we have a tuple $(c_a,w(a))$, with $c_a$ denoting to the number of times $a$ appears in the sample and $w(a)$ denotes the weight of element $a$. Since our computational model allows only the following two operations, given two sampled elements deciding whether they are the same and querying for the weight of a sampled element, it is sufficient to consider the fingerprints of the samples rather than the actual samples \cite{BT2024}.

To show the lower bound, we construct instances $I$ and $I'$ for which the fingerprints of a small sized sample are similar. In particular, we will show that, if we use only $m$ proportional samples whose fingerprint is denoted as $F_m$, the probability of observing the fingerprint $F_m$ from instance $I$ is nearly the same as the probability of observing it from instance $I'$. Let $\E$ and $\E'$ denote the events that instances $I$ and $I'$ are chosen, respectively. Given a (multi)set $S$ of $m$ samples, let $F_m$ denote its fingerprint. We define the following ratio $\mathcal{R}(F_m)$.

$$\mathcal{R}(F_m)=\frac{\Pr[F_m \mid \E']}{\Pr[F_m \mid \E]}$$

If this ratio is close to $1$, then we show that the Bayes classifier cannot distinguish between $\E$ and $\E'$ with probability at least $2/3$. Since the Bayes classifier minimizes the probability of a classification error, it follows that no algorithm using $m$ samples can distinguish between $I$ and $I'$ with probability at least $2/3$. A similar result has been proved as Lemma $13$ in \cite{BT2024}, and we give a short proof of this result for completeness.

\begin{claim}\label{claim:bayes} Given two disjoint events $\E$ and $\E'$ such that $\Pr[\E]=\Pr[\E']=1/2$, and a random variable $X$, let $\RR(X)=\frac{\Pr[X \mid \E']}{\Pr[X \mid \E]}$. If $\Pr[98/100\leq \RR(X)\leq 100/98]\geq 24/25$, then any classifier that takes $X$ as input, and classifies $\E$ and $\E'$, has error probability at least $1/3$. \end{claim}

\begin{proof} Let $q=\Pr[\E\mid X]$ and $q'=\Pr[\E'\mid X]$. Conditioned on $98/100\leq \RR(X)\leq 100/98$, we have $1/q=(q+q')/q=1+q'/q=1+\frac{\Pr[\E'\mid X]}{\Pr[\E\mid X]}=1+\frac{\Pr[X \mid \E']}{\Pr[X \mid \E]}=1+\RR(X)\in [198/100,198/98]$. Hence, $q\leq 100/198$. Similarly, $q'\leq 100/198$. The probability of a correct classification is $\Pr[\text{correct classification}~|~98/100\leq \RR(X)\leq 100/98]\cdot \Pr[98/100\leq \RR(X)\leq 100/98]+\Pr[\text{correct classification}~|~\RR(X)\notin [98/100,100/98]] \cdot \Pr[\RR(X)\notin [98/100,100/98]]\leq 1/25 + 100/198< 2/3$. \end{proof}

Let the instance $I(A,W)$ be defined as follows. The set $A$ of $n$ input elements can be written as $A = A_1 \uplus \ldots \uplus A_t$, where for any $i\in [t]$, $A_i\subseteq A$ denotes the subset of $n_i$ elements with weight $w_i$ and $t\in [n]$. We refer to each $A_i$ as a {\em block} in $A$. Note that $\sum_{i=1}^t n_i = n$ and $\sum_{i=1}^t n_i \cdot w_i = W$. Given such an instance $I(A, W)$, we create another instance $I'(A', W')$ as follows. For each $i\in [t]$, we include all elements of $A_i$ in $A'$, and add $\eps n_i$ new elements of weight $w_i$ in $A'$. Let $n_i'= n_i + \eps n_i$. The total number of elements in $A'$ is $\sum_{i=1}^t n_i'= (1+\eps)n$, and the total weight of elements in $A'$ is given as $W' = \sum_{i=1}^t w_i n_i' = (1+\eps)W$. Let $k_i=\frac{n_i p_i^2}{\norm{p}_2^2}$. Note that we allow the new elements to be added fractionally, when $\eps n_i$ is a fraction.\footnote{When $\eps n_i$ is a fraction, we assume an element can be added fractionally. The weight of the element remains $w_i$, however, the probability of sampling the element will be proportional to $\eps w_i$.} Next, we state the main claim that we prove in this section. 

\begin{thm}\label{lbmain} Let $0<\eps\leq 1$ be such that $\eps\geq 1/\sqrt{n_i}$ for all $i$ and let $t$ be a constant. Then, any $(\eps/3,1/3)$-estimate of $W$ requires $\Omega(\frac{1}{\eps \norm{p}_2})$ proportional samples. \end{thm}


We defined $\RR(F_m)$ as the ratio of the conditional probabilities of observing the fingerprint $F_m$ given that we sample $m$ elements from instances $I'$ and $I$, respectively. We want to bound the probability of obtaining a particular fingerprint $F_m$, where $F_m = \{(c_1, w(a_1)), (c_2, w(a_2)), \ldots, (c_l, w(a_l))\}$, where $l$ denotes the number of distinct elements in the sample. Let $m_i$ denote the number of sampled elements of weight $w_i$ and $l_i$ denote the number of distinct sampled elements with weight $w_i$. Let an element $a_j$ of weight $w_i$ be sampled $c_{ij}$ times. We have $\sum_{i=1}^t \sum_{j=1}^{l_i} c_{ij}=\sum_{i=1}^t m_i=m$, and $\sum_{i=1}^t l_i=l$. Let $p$ and $p'$ denote the proportional sampling distributions for instances $I$ and $I'$. The probabilities of sampling an element $a_i\in A$ of weight $w_i$ from $I$ and $I'$ are $p_i$ and $p'_i$, and hence, the probabilities of observing a sample of weight $w_i$ from $I$ and $I'$ are given as $n_i p_i$ and $n'_i p'_i$, respectively. Let $\Pr[F_m \mid \E]$ and $\Pr[F_m \mid \E']$ denote the conditional probabilities of observing fingerprint $F_m$ from $I$ and $I'$, respectively. Then, we have

$$ \Pr[F_m \mid \E] = \prod\limits_{i=1}^t \left((p_i)^{m_i} \binom{n_i}{l_i} \binom{m_i}{c_{i1} \ldots c_{il_i}} l_i!\right)$$  

$$ \Pr[F_m \mid \E'] = \prod\limits_{i=1}^{t} \left((p_i')^{m_i} \binom{n_i+\eps n_i}{l_i} \binom{m_i}{c_{i1} \ldots c_{il_i}} l_i!\right)$$  

\begin{claim}\label{claim:risk_as_product}  $\mathcal{R}(F_m)=\frac{\Pr[F_m \mid \E']}{\Pr[F_m \mid \E]}=\prod\limits_{i=1}^t\left(\left(\frac{1}{1+\eps}\right)^{m_i-l_i}\cdot\prod\limits_{j=0}^{l_i-1} \left(1+\frac{\eps j}{(1+\eps)(n_i-j)}\right)\right)$. \end{claim}
\begin{proof} 
\begin{equation*}
	\begin{aligned}
        	\mathcal{R}(F_m)=\frac{\Pr[F_m|\E']}{\Pr[F_m|\E]} 	& =\prod\limits_{i=1}^t\left(\frac{p_i'}{p_i}\right)^{m_i}\cdot \frac{\binom{n_i+\eps n_i}{l_i}}{\binom{n_i}{l_i}} \\
                							& = \prod\limits_{i=1}^t\left(\left(\frac{1}{1+\eps}\right)^{m_i}\cdot \prod\limits_{j=0}^{l_i-1}\frac{n_i+\eps n_i-j}{n_i-j}\right)\\
                							& =\prod\limits_{i=1}^t\left(\left(\frac{1}{1+\eps}\right)^{m_i}\cdot\prod\limits_{j=0}^{l_i-1}\left(1+\frac{\eps n_i}{n_i-j}\right)\right)\\
				                                        & = \prod\limits_{i=1}^t\left(\left(\frac{1}{1+\eps}\right)^{m_i}\cdot (1+\eps)^{l_i}\prod\limits_{j=0}^{l_i-1} \left(1+\frac{\eps j}{(1+\eps)(n_i-j)}\right)\right)\\
                                				        & = \prod\limits_{i=1}^t\left(\left(\frac{1}{1+\eps}\right)^{m_i-l_i}\cdot\prod\limits_{j=0}^{l_i-1} \left(1+\frac{\eps j}{(1+\eps)(n_i-j)}\right)\right)\\
            \end{aligned}
        \end{equation*}
\end{proof}


First, we show a bound on $m_i$. We show that with high probability either $m_i=o(\sqrt{n_i}/\eps)$ or $m_i\leq 1$. When $m_i=o(\sqrt{n_i}/\eps)$, using the assumption $\eps\geq 1/\sqrt{n_i}$, we get $m_i=o(n_i)$.

\begin{lem}\label{lem:m_i_not_large} Let $m=o(\frac{1}{\eps \norm{p}_2})$. Then, for all $i$, either $m_i=o(\frac{\sqrt{n_i}}{\epsilon})$ or $m_i<2$ with probability at least $1-\frac{1}{200t}$. \end{lem}

\begin{proof} Let $Z_{ik}$ be a random variable such that $Z_{ik}=1$ if the $k$th sample is from $A_i$ and $0$ otherwise. Then, $\Pr[Z_{ik}=1]=n_ip_i$. Let $m_i=\sum_{k=1}^m Z_{ik}$ denote the total number of samples from $A_i$. Let $\mu_i=\EE[m_i]=mn_ip_i$. We analyze two cases based on $\mu_i> \frac{1}{10\sqrt{t}}$ and $\mu_i\leq \frac{1}{10\sqrt{t}}$. For $\mu_i>\frac{1}{10\sqrt{t}}$ and $\delta=10\sqrt{t}\log (200t)$, using Chernoff bound, we get $$\Pr[m_i\geq (1+\delta)\mu_i]\leq \left(\frac{e^{\delta}}{(1+\delta)^{(1+\delta)}}\right)^{\mu_i}\leq \left(\frac{e}{\delta}\right)^{\delta \mu_i}=\frac{1}{\left(\frac{\delta}{e}\right)^{\delta \mu_i}}=\frac{1}{e^{\delta \mu_i\ln(\delta/e)}}\leq \frac{1}{e^{\delta \mu_i}}\leq \frac{1}{200t}$$

For $\mu_i\leq \frac{1}{10\sqrt{t}}$, we get $\Pr[m_i\geq 2]\leq \binom{m}{2}n_i^2 p_i^2\leq \frac{\mu_i^2}{2}\leq \frac{1}{200t}$. Since $t$ is a constant, we have $m_i\leq \max\{O(\mu_i),1\}$ with probability at least $1-1/200t$. For $m=o(\frac{1}{\eps \norm{p}_2})$, we have $\mu_i=mn_ip_i=o\left(\frac{n_ip_i}{\eps \norm{p}_2}\right)=o\left(\frac{n_i}{\eps}\cdot\sqrt{\frac{k_i}{n_i}}\right)=o\left(\frac{\sqrt{n_i}}{\eps}\right)$, where $k_i=\frac{n_i p_i^2}{\norm{p}_2^2}\leq 1$. Hence, $m_i\leq \max\{o(\sqrt{n_i}/\eps),1\}$ with probability at least $1-1/200t$.  

\end{proof}

We condition on the event $\E_1$ that Lemma \ref{lem:m_i_not_large} holds simultaneously for all $i$, and we have $\Pr[\E_1]\geq 1-1/200$. Conditioned on $\E_1$, we have either $m_i<2$ or $m_i=o(\sqrt{n_i}/\eps)=o(n_i)$, since we assume $\eps\geq 1/\sqrt{n_i}$ for all $i$. For a fixed value of $m_i$, using Claim \ref{claim:risk_as_product}, we can write $\mathcal{R}(F_m)=\prod\limits_{i=1}^t \mathcal{R}(l_i)$, where $\mathcal{R}(l_i)=\left(\left(\frac{1}{1+\eps}\right)^{m_i-l_i}\cdot\prod\limits_{j=0}^{l_i-1} \left(1+\frac{\eps j}{(1+\eps)(n_i-j)}\right)\right)$. In order to show that $\RR(F_m)$ is close to $1$, it is sufficient to prove that $\RR(l_i)$ is very close to $1$ with high probability for all $i=1,2,\ldots,t$, and then use a union bound. To show this, we first determine a range for each $l_i$. Specifically, we define an interval $[a_i,b_i]$ for $l_i$ such that $\Pr[a_i \leq l_i \leq b_i] > 1-\frac{1}{200t}$. Since $\RR(l_i)$ is an increasing function of $l_i$, it follows that $\Pr[\RR(a_i) \leq \RR(l_i) \leq \RR(b_i)] \geq 1-\frac{1}{200t}$. In the following, we obtain an interval for $l_i\in [a_i,b_i]$ conditioned on $m_i=r$, for some number $r$. Note that conditioned on $m_i=r$, the random variables corresponding to the samples from the $i$th block are independent. Before fixing the values of $a_i$ and $b_i$, we bound the variance of $l_i$. A similar result has been proved as Lemma $4.2$ of \cite{BT2024}. We state the lemma below.

\begin{lem}\label{lem:var-l} Let $m_i=r$ denote the number of sampled elements of weight $w_i$ from the instance $I$. Let there be $l_i$ distinct elements of weight $w_i$ in the sample. Then, $\var[l_i] \leq r^2/n_i$. Similarly, when $r$ elements of weight $w_i$ are sampled from $A'$, the variance of the number of distinct sampled elements of weight $w_i$ will be at most $\frac{r^2}{(1+\eps)n_i}$. \end{lem} 
\begin{proof} Let $Z_{ij}$ be a $\{0,1\}$-valued random variable that takes value $1$ if an element $a_j$ of weight $w_i$ appears in the sample and $0$ otherwise. Then, $l_i = \sum_{j=1}^{n_i} Z_{ij}$ equals the number of distinct elements of weight $w_i$ in the sample. We use the Efron-Stein inequality (Theorem \ref{thm:es}) to bound the variance of $l_i$. Let $X_i$ denote the $i$th sample. We write $l_i = f(X_1, X_2, \ldots, X_{r})$. Let $X'_1,X'_2,\ldots,X'_{r}$ be an independent copy of $X_1,X_2,\ldots,X_{r}$, and let $l'_{ik} = f(X_1, X_2, \ldots, X_{k-1}, X'_k, X_{k+1}, \ldots, X_{r})$ denote the number of distinct sampled elements of weight $w_i$, where the $k$th sample $X_k$ has been resampled as $X'_k$. Then, we have

$$ l_i - l_{ik}' =      \begin{cases}
                                1 & X_k~\text{does not collide with any}~X_j~\text{for}~(j \neq k)~\text{and}~X_k'~\text{collides with}~X_{j'}~\text{for}~(j'\neq k), \\
                                0 & X_k, X_k'~\text{both collide or both are unique}, \\
                                -1 & X_k'~\text{does not collide with any}~X_j~\text{for}~(j\neq k)~\text{and}~X_k~\text{collides with}~X_{j'}~\text{for}~(j'\neq k).
                        \end{cases} $$

Hence, $\EE[(l_i - l_{ik}')^2] = \Pr[(l_i - l_{ik}')^2 = 1] = \Pr[(l_i - l_{ik}') = 1] + \Pr[(l_i - l_{ik}') = -1]$. Now,

\begin{equation*}
\begin{aligned}
& \Pr[(l_i - l_{ik}') = 1]\\
                          & =  \Pr\left[(X_k~\text{does not collide with any}~X_j~\text{for}~(j\neq k))~\cap~(X_k'~\text{collides with}~X_{j'}~\text{for}~(j'\neq k))\right] \\
                          & =  \Pr\left[X_k~\text{does not collide with any}~X_j~\text{for}~(j\neq k)\right] \cdot \Pr\left[X_k'~\text{collides with}~X_{j'}~\text{for}~(j'\neq k)\right] \\
                          & \leq 1 \cdot \Pr\left[\cup_{j'\neq k} (X_k'~\text{collides with}~X_{j'})\right] \\
                          & \leq \sum_{j' \neq k} \Pr[X_k'~\text{collides with}~X_{j'}] \\
                          & =  \sum_{j'\neq k}\Pr[\cup_{a\in A_i}(X_k'=X_{j'}=a)] \\
                          & \leq  \sum_{j'\neq k} \sum_{a\in A_i}\Pr(X_k'=X_{j'}=a)\\
                          & =  (r - 1)\cdot n_i\cdot 1/n_i^2 \\
                          & \leq r/n_i.
\end{aligned}
\end{equation*}

Similarly, we can show that $\Pr[(l_i - l_{ik}') = -1] \leq r/n_i$, and hence $\EE[(l_i - l_{ik}')^2] \leq 2\cdot r/n_i$. Then, by the Efron-Stein inequality, we get $ \var[l_i] \leq \frac{1}{2}\sum_{k=1}^{r} \EE[(l_i - l_{ik}')^2] \leq \frac{1}{2}\sum_{k=1}^{r} 2r/n_i = r^2/n_i$. \end{proof}

Next, we bound the difference between the expected number of distinct sampled elements when we sample $m_i=r$ elements of weight $w_i$ from $I$ and $I'$, respectively. We show that, conditioned on $m_i=r$, the difference between the expected number of distinct sampled elements when we sample $m_i$ elements of weight $w_i$ from instances $I$ and $I'$ is not large. We state the lemma below. 

\begin{lem}\label{lem:exp-l} Let $m_i$ denote the number of sampled elements of weight $w_i$, and let $l_i$ denote the number of distinct elements of weight $w_i$ in the sample. Then, $$\EE_{\E'}[l_i \mid m_i=r] - \EE_{\E}[l_i \mid m_i=r] \leq O\left(\eps \frac{r^2}{n_i}\right). $$ \end{lem}

\begin{proof} First, we bound the expected number of distinct sampled elements of weight $w_i$ when we sample $m_i$ elements from $I$. Let $Z_{ij}$ be a $\{0,1\}$-valued random variable that takes value $1$ if an element $a_j$ of weight $w_i$ appears in the sample and $0$ otherwise. Then, $l_i = \sum_{j=1}^{n_i} Z_{ij}$ equals the total number of distinct elements of weight $w_i$ in the sample. Conditioned on sampling $r$ elements from $A_i$, the probability of a particular element $a_j$ in $A_i$ is given as $(1-(1-1/n_i)^r)$. Hence, $\EE_{\E}[l_i \mid m_i=r]=n_i \left(1-(1-\frac{1}{n_i})^{r}\right)$. Using the binomial expansion, we get

$$ \EE_{\E}[l_i \mid m_i=r]=n_i \left(1-(1-\frac{1}{n_i})^{r}\right) = n_i \left(\frac{r}{n_i} - \binom{r}{2} \frac{1}{n_i^2}+ \binom{r}{3} \frac{1}{n_i^3} - \cdots + (-1)^{r+1} \frac{1}{n_i^r}\right).$$

Next, we obtain a similar expression for the expected number of distinct sampled elements of weight $w_i$ when we sample $r$ elements from $A'_i$. We have $\EE_{\E'}[l_i \mid m_i=r]=n_i(1+\eps) \left(1-(1-\frac{1}{n_i(1+\eps)})^{r}\right)$. Using the binomial expansion again, we obtain

\begin{equation*}
    \begin{aligned}
        \EE_{\E'}[l_i \mid m_i=r] &=n_i(1+\eps) \left(1-(1-\frac{1}{n_i(1+\eps)})^{r}\right)\\ 
        &= n_i(1+\eps) \left(\frac{r}{n_i(1+\eps)} - \binom{r}{2} \frac{1}{n_i^2(1+\eps)^2}+ \binom{r}{3} \frac{1}{n_i^3(1+\eps)^3} - \cdots + (-1)^{r+1} \frac{1}{n_i^r(1+\eps)^r}\right)
    \end{aligned}
\end{equation*}

We can write,
\begin{equation*}
    \begin{aligned}
        &\EE_{\E'}[l_i \mid m_i=r] - \EE_{\E}[l_i \mid m_i=r] \\
        & = (n_i+\epsilon n_i)\left(1-(1-\frac{1}{(1+\eps)n_i})^{r}\right)-n_i\left(1-(1-\frac{1}{n_i})^{r}\right)\\
        &= n_i(1+\epsilon) \left(\frac{r}{n_i(1+\epsilon)} - \binom{r}{2} \frac{1}{n_i^2(1+\epsilon)^2}+ \binom{r}{3} \frac{1}{n_i^3(1+\epsilon)^3} - \cdots + (-1)^{r+1} \frac{1}{n_i^r(1+\epsilon)^r}\right) \\
                & \hspace{0.5cm}- . n_i \left(\frac{r}{n_i} - \binom{r}{2} \frac{1}{n_i^2}+ \binom{r}{3} \frac{1}{n_i^3} - \cdots + (-1)^{r+1} \frac{1}{n_i^r}\right)\\
            & = n_i\cdot \sum_{t=2}^{r}(-1)^{t+1}\binom{r}{t}\frac{1}{n_i^t}\left(\frac{1}{(1+\epsilon)^{t-1}}-1\right)\\
            & \leq n_i\cdot \sum_{t=2}^{r}(-1)^{t+1}\binom{r}{t}\frac{1}{n_i^t}\left(\frac{-\epsilon}{2}\right) & \hspace{-4cm}\hfill [\text{since}~ \forall t\geq 2, \left(\frac{1}{(1+\epsilon)^{t-1}}-1\right)\leq -\frac{\epsilon}{2} ]\\
            &=  \frac{\epsilon n_i}{2}\cdot \sum_{t=2}^{r}(-1)^{t}\binom{r}{t}\frac{1}{n_i^t}\\
        & = \frac{\epsilon n_i}{2}\cdot \binom{r}{2}\frac{1}{n_i^2}+\frac{\epsilon n_i}{2}\cdot \sum_{t=3}^{r}(-1)^{t}\binom{r}{t}\frac{1}{n_i^t} \\
        & \leq \frac{\epsilon n_i r^2}{4n_i^2}+ \frac{\epsilon n_i}{2}\cdot \sum_{t=3}^{r}\binom{r}{t}\frac{1}{n_i^t}\\
        & \leq \frac{\epsilon r^2}{4 n_i}+ \frac{\epsilon n_i}{2}\cdot \sum_{t=3}^{\infty}\frac{r^t}{t!}\frac{1}{n_i^t}\\
        & = \frac{\epsilon r^2}{4 n_i}+ \frac{\epsilon n_i}{2}\cdot \left(e^{r/n_i} - 1 - \frac{r}{n_i} - \frac{r^2}{2n_i^2}\right)\\
        & = O\left(\frac{\epsilon r^2}{n_i}\right)+O\left(\epsilon n_i \frac{ r^3}{n_i^3}\right) & \hspace{-7cm} \hfill [\text{since}~\forall x<1,\; e^x-1-x-x^2/2=O(x^3)]\\
        & = O\left(\frac{\epsilon r^2}{n_i}\right) & \hspace{-1cm}\hfill [\text{since}~r=o(n_i)]
    \end{aligned}
\end{equation*}

\end{proof}

Let us set $a_i=\EE_{\E}[l_i|m_i=r]-20\sqrt{t} \cdot r/\sqrt{n_i}$ and $b_i=\EE_{\E'}[l_i|m_i=r]+20\sqrt{t} \cdot r/\sqrt{n_i}$, and show that the value of $l_i$ lies in this range $[a_i,b_i]$ with probability at least $1-\frac{1}{200t}$.
        \begin{equation*}
                \begin{aligned}
                        |b_i-a_i| &= \EE[l_i|\E']-\EE[l_i|\E]+40 \sqrt{t} \cdot r/\sqrt{n_i}\\
                         & \leq O(\eps r^2/n_i) + 40\sqrt{t} \cdot r/\sqrt{n_i}   &(\text{Using Lemma \ref{lem:exp-l}})\\
                \end{aligned}
        \end{equation*}

Using the above values for $a_i$ and $b_i$, we bound the probability that $l_i$ does not belong to the interval $[a_i,b_i]$.

        \begin{equation*}
                \begin{aligned}
                        \Pr[l_i>b_i] & = \Pr[l_i>\EE[l_i|\E']+20\sqrt{t} \cdot r/\sqrt{n_i}] \\
                & = \Pr[l_i-\EE[l_i|\E']> 20\sqrt{t} \cdot r/\sqrt{n_i}]\\
                & \leq \frac{\var[l_i]}{400t \cdot r^2/n_i}\\
                & \leq \frac{r^2/(1+\eps)n_i}{400t r^2/n_i} &(\text{Using Lemma \ref{lem:var-l}})\\
                & \leq \frac{1}{400t} 
        \end{aligned}
\end{equation*}

Similarly, we can show that $\Pr[l_i<a_i]\leq \frac{1}{400t}$. Hence, $\Pr[l_i\notin [a_i,b_i]]\leq \frac{1}{200t}$. Hence, $\Pr[l_i\in [a_i,b_i]]\geq 1-1/200t$. Now, $|b_i-a_i|=\EE_{\E'}[l_i|m_i=r]-\EE_{\E}[l_i|m_i=r]+40\sqrt{t} \cdot r/\sqrt{n_i}$. Let $\E_2$ denote the event that $l_i\in [a_i,b_i]$ for all $i$, and using a union bound, we have $\Pr[\E_2]\geq 1-1/200$. We condition on event $\E_2$ for the following. Since $a_i<l_i<b_i$, and $\RR(x)$ is an increasing function of $x$, we have $\RR(a_i)<\RR(l_i)<\RR(b_i)$ with probability at least $1-\frac{1}{200t}\geq e^{-1/100t}$, where we use the inequality $e^{-x}\leq 1-x/2$ for $x\in [0,1.59]$. First, we obtain an upper bound on $\mathcal{R}(b_i)$ and a lower bound on $\mathcal{R}(a_i)$. Since $\RR(l_i)$ is an increasing function of $l_i$, we can write 

$$\mathcal{R}(a_i) (1-1/200t) \leq \mathcal{R}(a_i) \sum_{j\in [a_i,b_i]} \Pr[l_i=j|\E] \leq \sum_{j\in [a_i,b_i]} \mathcal{R}(j) \Pr[l_i=j|\E] \leq \sum_{j\in [a_i,b_i]} \Pr[l_i=j|\E'] \leq 1$$

Hence, $\RR(a_i) \cdot e^{-1/100t}\leq \RR(a_i) \cdot (1-1/200t)\leq 1$, and we have $\RR(a_i)\leq e^{1/100t}$. Similarly,
        
$$ \frac{1}{\mathcal{R}(b_i)} (1-1/200t) \leq \frac{1}{\mathcal{R}(b_i)} \sum_{j\in [a_i,b_i]} \Pr[l_i=j|\E'] \leq \sum_{j\in [a_i,b_i]} \frac{1}{\mathcal{R}(j)} \Pr[l=j|\E'] \leq \sum_{j\in [a_i,b_i]} \Pr[l_i=j|\E] \leq 1$$

Hence, $\mathcal{R}(b_i) \geq (1-1/200t)\geq e^{-1/100t}$. At this point, we have an upper bound on $\mathcal{R}(a_i)$ and a lower bound on $\mathcal{R}(b_i)$. We need a lower bound on $\mathcal{R}(a_i)$ and an upper bound on $\mathcal{R}(b_i)$ for our analysis. We first obtain an upper bound on $\prod_{l_i=a_i+1}^{b_i} \frac{\mathcal{R}(l_i)}{\mathcal{R}(l_i-1)}$.

\begin{lem}\label{lem:rel_a_i_b_i} $\prod\limits_{l_i=a_i+1}^{b_i}\frac{\mathcal{R}(l_i)}{\mathcal{R}(l_i-1)}\leq e^{2\eps\cdot |b_i-a_i|}$. \end{lem}

\begin{proof} We define $\mathcal{R}(l_i)=\left(\left(\frac{1}{1+\eps}\right)^{m_i-l_i}\cdot\prod\limits_{j=0}^{l_i-1} \left(1+\frac{\eps j}{(1+\eps)(n_i-j)}\right)\right)$, and we can write $\frac{\mathcal{R}(l_i)}{\mathcal{R}(l_i-1)}=(1+\eps)\left(1+\frac{\eps (l_i-1)}{(1+\eps)(n_i-l_i+1)}\right)=1+\eps +\frac{\eps (l_i-1)}{(n_i-l_i+1)}\leq 1+\eps + \frac{\eps m_i}{(n_i-m_i+1)}$. Conditioned on $\E_1$, we have either $m_i=o\left(\frac{\sqrt{n_i}}{\epsilon}\right)$ or $m_i<2$. When $m_i=o\left(\frac{\sqrt{n_i}}{\epsilon}\right)$, in conjunction with $\eps>\frac{1}{\sqrt{n_i}}$, we have $m_i=o(n_i)$. In either of the two cases above, we obtain the following bound. 

$$\frac{\mathcal{R}(l_i)}{\mathcal{R}(l_i-1)}\leq 1+\eps +\frac{\eps m_i}{(n_i-m_i+1)}\leq 1+\eps +\eps =1+2\eps\leq e^{2\epsilon}.$$ 

Hence, we have $\prod\limits_{l_i=a_i+1}^{b_i}\frac{\mathcal{R}(l_i)}{\mathcal{R}(l_i-1)}\leq \prod\limits_{l_i=a_i+1}^{b_i} e^{2\epsilon}\leq e^{2\eps |b_i-a_i|}$. 
\end{proof}
Using Lemma \ref{lem:rel_a_i_b_i}, we obtain $\RR(b_i)\leq e^{1/100t} \cdot e^{2\eps |b_i-a_i|}$ and $\RR(a_i)\geq e^{-1/100t} \cdot e^{-2\eps |b_i-a_i|}$. Conditioned on $\E_2$, and using $\mathcal{R}(F_m)=\prod\limits_{i=1}^t \mathcal{R}(l_i)$, we can write 
$$\prod\limits_{i=1}^t e^{-\frac{1}{100t}} \cdot e^{-2\eps |b_i-a_i|} \leq \RR(F_m) \leq  \prod\limits_{i=1}^t e^\frac{1}{100t} \cdot e^{2\eps |b_i-a_i|}$$
This can be simplified as follows.
\begin{equation} \label{eq:main} e^{-\frac{1}{100}\sum\limits_{i=1}^t 1/t} \cdot e^{-2\eps\sum\limits_{i=1}^t |b_i-a_i|} \leq \mathcal{R}(F_m) \leq e^{\frac{1}{100}\sum\limits_{i=1}^t 1/t} \cdot e^{2\eps\sum\limits_{i=1}^t |b_i-a_i|}.\end{equation}
We obtain an upper bound on $\sum\limits_{i=1}^t |b_i-a_i|$. Conditioned on $\E_1$, we have either $m_i\leq 1$ or $m_i=o(\sqrt{n_i}/\eps)$. When $m_i\leq 1$, we have $|b_i-a_i|=0$. For the following calculations, we consider blocks for which $m_i=O(\mu_i)=O(mn_ip_i)=o(\sqrt{n_i}/\eps)$, for $m=o(\frac{1}{\eps \norm{p}_2})$.
\begin{equation*}
    \begin{aligned}
        \sum_{i=1}^t |b_i-a_i|  &= \sum_{i=1}^t (\EE_{\E'}[l_i|m_i=r]-\EE_{\E}[l_i|m_i=r]) + \sum_{i=1}^t 40\sqrt{t} \cdot r/\sqrt{n_i}\\
        & = \sum_{i=1}^t O(\eps r^2/n_i) + 40\sqrt{t} \sum_{i=1}^t r/\sqrt{n_i} \qquad\qquad \quad(\text{using Lemma}~\ref{lem:exp-l})\\
        & = \sum_{i=1}^t O(\frac{\eps \mu_i^2}{n_i}) + 40\sqrt{t} \sum_{i=1}^t O(\mu_i/\sqrt{n_i})  \qquad \qquad (\text{Since}~r\leq O(\mu_i)~\text{using Lemma}~\ref{lem:m_i_not_large})\\
        & = \sum_{i=1}^t O(\eps m^2n_ip_i^2) + 40\sqrt{t}  \sum_{i=1}^t O(m\sqrt{n_i}p_i)\\
	&\leq C_1 \eps m^2 \sum_{i=1}^t n_i p_i^2 + 40C_2 m\sqrt{t} \sqrt{t} \sqrt{\sum_{i=1}^t n_i p_i^2} \quad \quad (\text{Cauchy-Schwarz inequality})\\
	&= C_1 \eps m^2 \norm{p}_2^2 + 40C_2 m t \norm{p}_2\\
        & = o(1/\eps)  \qquad \qquad \qquad \qquad \qquad \qquad \qquad \qquad \qquad(\text{For}~t=O(1), m=o(\frac{1}{\eps \norm{p}_2}))
    \end{aligned}
\end{equation*}

Using $\sum\limits_{i=1}^t |b_i-a_i|=o(1/\eps)$, we simplify Equation \ref{eq:main} as follows. $$e^{-\frac{1}{100}}\cdot e^{-o(1)} \leq \mathcal{R}(F_m) \leq e^{\frac{1}{100}}\cdot e^{o(1)}.$$

Simplifying, we get $e^{-\frac{2}{100}} \leq \mathcal{R}(F_m) \leq e^{\frac{2}{100}}$. This gives $\frac{98}{100}\leq \mathcal{R}(F_m)\leq \frac{100}{98}$. We bound the failure probability of the event $\mathcal{R}(F_m)\notin [98/100,100/98]$ using a union bound on the events $\E_1$ and $\E_2$. The failure probability is at most $1/200+1/200\leq 1/100$. Therefore, $\Pr[\mathcal{R}(F_m)\in [98/100,100/98]]\geq \frac{99}{100}$. Using Claim \ref{claim:bayes}, it follows that the Bayes classifier cannot succeed with probability at least $2/3$ in distinguishing whether the samples are from $I$ or $I'$. Hence, any algorithm that estimates of $W$ within a $(1\pm \eps)$ factor with probability at least $2/3$ requires $\Omega(\frac{1}{\eps \norm{p}_2})$ proportional samples on any instance $I(A,W)$. 

\subsection{Lower bound for large $\norm{p}_2$} \label{sec:lb-large} We show a $\Omega(\frac{1}{\eps^2 \norm{p}_2})$ lower bound for the sum estimation problem, when $\norm{p}_2$ is at least a constant and $\eps$ is not too small. We state the main claim below.
\begin{claim}\label{claim:large-norm} For $\eps>1/\sqrt{n}$, there exist instances $I(A,W)$ with $\norm{p}_2$ being a constant for which any $(\eps,1/3)$ estimator of the sum of the weights requires $\Omega(\frac{1}{\eps^2 \norm{p}_2})$ samples. \end{claim}

\begin{proof} Let instance $I(A,W)$ have an element $a_0$ of weight $w_1$ and $(n-1)$ elements of weight $w_2$. We have $A=A_1\cup A_2$, where $A_1=\{a_0\}$. Let $n_1=|A_1|=1$ and $n_1 \cdot w_1= c_1 W$, for some constant $0<c_1<1$, and $n_2=|A_2|=n-1$ and $(n-1)w_2=c_2 W$, for some constant $c_2=1-c_1>0$. We can also write $n_1 p_1 = c_1$ and $n_2 p_2 = c_2$. Given such an instance $I(A, W)$, we create another instance $I'(A',W')$ as follows. We include $a_0$ and all elements of $A_2$ in $A'$ and add $\eps_2 n_2$ new elements of weight $w_2$ in $A'$, where we set $\eps_2=\frac{\eps}{c_2}$. Note that the total weight of elements in instance $I$ is $W$ and total weight of elements in $I'$ is $W'= W+\eps_2 n_2 w_2=W+\frac{\eps}{c_2}\cdot c_2 W=W+\eps W= (1+\eps)W$. As before, we assume the adversary chooses instance $I$ or $I'$ with equal probability and we have $\Pr[\E]=\Pr[\E']=1/2$. We define $q=\Pr[a\neq a_0|\E]$ and $q'=\Pr[a\neq a_0|\E']$.

\begin{align*} 
        q &= \frac{c_2 W}{W}=c_2\\
        q' &= \frac{c_2 W+\eps W}{W+\eps W}\leq \frac{c_2W+\eps W}{W}= q+\eps
\end{align*} 

This gives $|q-q'|\leq \eps$. Now, we define a random variable
    
\[X_i   =       \begin{cases}
                        1 & \text{Weight of sample is not equal to}~w_1 \\ 0 & \text{otherwise} 
                \end{cases}
\]

Then, $X=\sum_{i=1}^m X_i$ denotes the number of sampled elements with weight not equal to $w_1$. We define the ratio $\mathcal{R}(X)=\frac{\Pr[X|\E']}{\Pr[X|\E]}$. Since $q$ is a constant, we can write $m=o(1/\eps^2)=o(q/\eps^2)$. Then, using Lemma \ref{lem:44BT2024}, we get $\Pr\left[\frac{98}{100}\leq\mathcal{R}(X) \leq \frac{100}{98} \right]\geq \frac{99}{100}$. Let $F_X$ denote the fingerprint of the sampled elements other than $a_0$. From the given $F_X$, we can reconstruct the fingerprint of the whole sample $F_m$ by adding the tuple $(m-X,w(a_0))$. Hence, $\Pr[F_m|\E]=\Pr[F_X|\E]$ and $\Pr[F_m|\E']=\Pr[F_X|\E']$. Then,

\[\mathcal{R}(F_m)=\frac{\Pr[F_m|\E']}{\Pr[F_m|\E]}=\frac{\Pr[F_X|\E']}{\Pr[F_X|\E]}=\frac{\Pr[X|\E']\cdot \Pr[F_X|X,\E']}{\Pr[X|\E] \cdot \Pr[F_X|X,\E]}=\mathcal{R}(X)\cdot \mathcal{R}(F_X|X)\]

We compute $\mathcal{R}(F_X|X)$ conditioned on the event $\mathcal{L}=\{X\leq 30\EE[X]\}$. By Markov's inequality, we get $\Pr[\mathcal{L}]\geq 29/30$. Since $\Pr[\E]=\Pr[\E']=1/2$, we have 

\[\EE[X]=m(q+q')/2\leq m(q+\eps)\leq m\cdot \left(c_2+\eps\right)\leq 2m=o\left(\frac{1}{\eps^2}\right)=o\left(\frac{\sqrt{n}}{\eps}\right)\]

We have $X\leq 30 \EE[X]=o\left(\frac{\sqrt{n}}{\eps}\right)=o\left(\frac{\sqrt{n_2}}{\eps_2}\right)$. Conditioning on $X=x$ for some $x\leq 30\EE[X]$ and using Lemma \ref{lem:43BT2024}, we get

\[\Pr\left[\frac{98}{100}\leq \frac{\Pr[F_X|X=x,\E']}{\Pr[F_X|X=x,\E]}\leq \frac{100}{98}\right]\geq \frac{99}{100}\]

and integrating over $\mathcal{L}$, we obtain

\[\Pr\left[\frac{98}{100}\leq \frac{\Pr[F_X|X,\E']}{\Pr[F_X|X,\E]}\leq \frac{100}{98}\Bigg| \mathcal{L}\right]\geq \frac{99}{100}.\]

Then, conditioned on $\mathcal{L}$, we get $96/100\leq (98/100)^2\leq \mathcal{R}(F_m)\leq (100/98)^2\leq 100/96$. Then, using a union bound, the total failure probability can be bounded as follows.

\begin{equation*}
    \begin{aligned}
        \Pr\left[\mathcal{R}(F_m)\notin \left[\frac{96}{100},\frac{100}{96}\right]\right] & \leq \Pr[\mathcal{L}^c]+\Pr\left[\frac{\Pr[F_X|X,\E']}{\Pr[F_X|X,\E]}\notin \left[\frac{98}{100},\frac{100}{98}\right]\Bigg| \mathcal{L}\right]\\ & \hspace{2cm}+\Pr\left[\frac{\Pr[X|\E']}{\Pr[X|\E]}\notin \left[\frac{98}{100},\frac{100}{98}\right] \right] \\
        & \leq \frac{1}{30}+\frac{1}{100}+\frac{1}{100}\\
        & \leq \frac{6}{100}
    \end{aligned}
\end{equation*}

Then, $\Pr[96/100\leq \RR(F_m)\leq 100/96]\geq 94/100$. Using a similar argument as Claim \ref{claim:bayes}, we can say that the Bayes classifier taking $F_m$ and classifying $\E$ and $\E'$ has probability of error at least $1/3$. Hence, the Bayes classifier does not succeed with probability at least $2/3$ in distinguishing the instances. Hence, any algorithm that estimates $W$ within a $(1\pm \eps)$ factor with probability at least $2/3$ requires $\Omega\left(\frac{1}{\eps^2}\right)$ proportional samples. Since $p_1$ is a constant, $\norm{p}_2$ is also a constant. Hence, we can write lower bound as $\Omega\left(\frac{1}{\eps^2 \norm{p}_2}\right)$.
\end{proof}

\section{Conclusions} We provide an improved algorithm for sum estimation using proportional sampling that on an instance with proportional sampling distribution $p$, computes an $(\eps,\delta)$-estimate of the sum of the weights of the elements using $O((\frac{1}{\eps^2 \norm{p}_2}+\frac{\ln 1/\norm{p}_2}{\norm{p}_2})\log 1/\delta)$ proportional samples. We give an almost matching lower bound of $\Omega\left(\frac{1}{\eps \norm{p}_2}\right)$ proportional samples for any $(\eps,1/3)$-sum estimation problem, assuming for all $i\in [t]$, $\eps\geq 1/\sqrt{n_i}$ and $t$ is a constant. We also provide an algorithm that computes an $(\eps,\delta)$-estimate of the $h$-th moment of the weights of the elements using $O((\frac{h}{\norm{p}_h}+\frac{1}{\eps^2 \norm{p}_2}+\frac{\ln 1/\norm{p}_2}{\norm{p}_2}+\frac{1}{\eps^2 \norm{p}_h})\log 1/\delta)$ proportional samples. 

\section{Acknowledgment}
 Anup Bhattacharya acknowledges the support from the SERB, CRG grant (CRG/2023/002119).

\section{AI Usage disclosure} All the main ideas of the paper were generated by the authors. The paper is written entirely by the authors. We used Gemini 3.1 Pro and ChatGPT 5.6 Luna to go over our draft and report any inconsistencies. The responses were used to improve the writing of the paper, in particular, some parts of our sum estimation algorithm and the lower bound.

\bibliographystyle{alpha}
\bibliography{ref}


\end{document}